\documentclass[journal]{IEEEtran}
\usepackage{amsfonts}
\usepackage{amssymb}
\usepackage{amsthm}
\usepackage{amsmath,amsfonts,amssymb}
\usepackage[dvips]{graphicx}
\usepackage{verbatim}
\usepackage{setspace}
\usepackage{bm}
\usepackage{algorithmic} %format of the algorithm
\usepackage[ruled,vlined]{algorithm2e}
\usepackage{cite}

\usepackage{changepage}
\usepackage{pdfpages}
\usepackage{color}
\newtheorem{theorem}{Theorem}
\newtheorem{lemma}{Lemma}

\newcommand{\figwidth}{8}
\IEEEoverridecommandlockouts

\begin{document}
\title{Joint Tx-Rx Beamforming and Power Allocation for 5G Millimeter-Wave Non-Orthogonal Multiple Access (MmWave-NOMA) Networks}

\author{Lipeng Zhu,
        Jun Zhang,
        Zhenyu Xiao,~\IEEEmembership{Senior Member,~IEEE,}
        Xianbin Cao,~\IEEEmembership{Senior Member,~IEEE,}
        Dapeng Oliver Wu,~\IEEEmembership{Fellow,~IEEE}
        and Xiang-Gen Xia,~\IEEEmembership{Fellow,~IEEE}
%\thanks{This work was supported in part by the National Natural Science Foundation of China (NSFC) under grant Nos. 61571025 and 91538204, and in part by the Foundation for Innovative Research Groups of the National Natural Science Foundation of China under grant No. 61521091.}
\thanks{L. Zhu, J. Zhang, Z. Xiao and X. Cao are with the School of
Electronic and Information Engineering, Beihang University, Beijing 100191, China. \{zhulipeng@buaa.edu.cn, buaazhangjun@vip.sina.com, xiaozy@buaa.edu.cn, xbcao@buaa.edu.cn\}.}
\thanks{D. O. Wu is with the Department of Electrical and Computer Engineering, University of Florida, Gainesville, FL 32611, USA. \{dpwu@ufl.edu\}.}
\thanks{X.-G. Xia is with the Department of Electrical and Computer Engineering, University of Delaware, Newark, DE 19716, USA. \{xianggen@udel.edu\}.}
}


\maketitle

\begin{abstract}
In this paper, we investigate the combination of non-orthogonal multiple access and millimeter-Wave communications (mmWave-NOMA). A downlink cellular system is considered, where an analog phased array is equipped at both the base station and users. A joint Tx-Rx beamforming and power allocation problem is formulated to maximize the achievable sum rate (ASR) subject to a minimum rate constraint for each user. As the problem is non-convex, we propose a sub-optimal solution with three stages. In the first stage, the optimal power allocation with a closed form is obtained for an arbitrary fixed Tx-Rx beamforming. In the second stage, the optimal Rx beamforming with a closed form is designed for an arbitrary fixed Tx beamforming. In the third stage, the original problem is reduced to a Tx beamforming problem by using the previous results, and a boundary-compressed particle swarm optimization (BC-PSO) algorithm is proposed to obtain a sub-optimal solution. Extensive performance evaluations are conducted to verify the rational of the proposed solution, and the results show that the proposed sub-optimal solution can achieve a near-upper-bound performance in terms of ASR, which is significantly improved compared with those of the state-of-the-art schemes and the conventional mmWave orthogonal multiple access (mmWave-OMA) system.
\end{abstract}


% \fi
\IEEEpeerreviewmaketitle

\section{Introduction}
\IEEEPARstart{W}{ith} the coming of the fifth-generation (5G) mobile communication, the urgent requirements of high spectrum efficiency, low latency, low cost and massive connectivity pose great challenges \cite{andrews2014will,Ding2017survNOMA,Dai2018NOMAsurvey,XiaoM2017survmmWave}. The conventional orthogonal multiple access (OMA) techniques, i.e., frequency division multiple access (FDMA), time division multiple access (TDMA), code division multiple access (CDMA) and orthogonal frequency division multiple access (OFDMA) may not meet the insistent requirements of mobile Internet and Internet of things (IoT) due to massive connectivities. Recently, non-orthogonal multiple access (NOMA) has become a promising candidate technology for 5G \cite{Benjebbour2013ConceptNOMA,Dai2015NOMA5G,andrews2014will,Ding2017survNOMA,Dai2018NOMAsurvey,Choi2014NOMA}. The basic idea of NOMA is to serve multiple users in an orthogonal frequency/time/code resource block and distinguish them in power domain. The signals of the users are decoded by using successive interference cancellation (SIC) at the receivers. In general, the users are sorted by their channel gains, and the decoding order is the increasing order of the channel gains. The first user, which has the lowest channel gain, directly decodes its signal treating the other signals as noise. The second user decodes the signal of the first user and removes it from the received signals. Then, it decodes its signal treating the other signals as noise. So on and so forth, the last user decodes and removes all the other users' signals before decoding its signal. In such a manner, the users with different channel conditions can transmit simultaneously. Both the spectrum efficiency and the number of users can be increased.

In addition to NOMA, millimeter-Wave (mmWave) communications is another enabling candidate technology for 5G \cite{andrews2014will,niu2015survey,rapp2013mmIEEEAccess,XiaoM2017survmmWave,xiao2017mmWaveFD,andrews2016modeling}. The abundant bandwidth of mmWave, from 30 GHz to 300 GHz, can significantly enhance the system capacity. Different from the existing microwave-band systems, the propagation in the mmWave-band presents high path loss and angle-domain spatial sparsity. Profiting from the short wavelength of the mmWave signal, large scale antenna array can be equipped in a small area. Thus, directional beamforming is usually utilized to obtain the beam gain, which can compensate the path loss in the mmWave band.

The combination of NOMA and mmWave communications (mmWave-NOMA) has some unique advantages. On one hand, the highly directional feature of mmWave transmission implies that the users' channels can be highly correlated, which is appropriate for applying NOMA. On the other hand, due to the high hardware consumption in the mmWave band, the number of radio frequency (RF) chains is limited in general. Thus, applying NOMA in mmWave communications can significantly increase the number of users \cite{Ding2017random,Cui2018mmWaveNOMA,Daill2017}.

Motivated by these advantages, we investigate mmWave-NOMA in this paper. There are several prior works on mmWave-NOMA in the existing literatures. Random beamforming approaches for a single-beam case and for a limited feedback case were proposed for mmWave-NOMA in \cite{Ding2017random}. The simulation results proved that mmWave-NOMA could achieve significant gains in terms of sum rates and outage probabilities, compared with the conventional mmWave-OMA. In \cite{Cui2018mmWaveNOMA}, the authors exploited beamforming, user scheduling and power allocation in mmWave-NOMA networks. By invoking random beamforming, the global optimal solution of power allocation can be obtained by the proposed branch and bound approach. Then, a low-complexity suboptimal approach was developed for striking a good computational complexity-optimality tradeoff. In \cite{Daill2017}, a new transmission scheme of beamspace multiple-input multiple-output NOMA (MIMO-NOMA) was proposed, where the number of the users can be larger than the number of the RF chains. A hybrid precoding approach based on the equivalent channel vector was proposed first. Then, an iterative algorithm was developed to optimize the power allocation for the users. In \cite{Wu2017hybridBF}, a NOMA based hybrid beamforming design in mmWave systems was considered. A low complexity user pairing algorithm was proposed first to reduce the interferences. Then, the authors proposed a hybrid beamforming algorithm with low feedback. Finally, the power allocation scheme was proposed to maximize the sum capacity. An in-depth capacity analysis for the integrated NOMA-mmWave-massive-MIMO systems was provided in \cite{Zhang2017mmWaveMIMONOMA}. A simplified mmWave channel model was proposed first. Whereafter, theoretical analysis on the achievable capacity was considered in both the low signal to noise ratio (SNR) and high-SNR regimes based on the dominant factors of signal to interference plus noise ratio. In \cite{xiao2018mmWaveNOMA}, a non-convex optimization problem was formulated and solved in a 2-user downlink mmWave-NOMA system, where power allocation and Tx beamforming were jointly optimized to maximize the achievable sum rate (ASR). Then, a joint power control and Rx beamforming problem in a 2-user uplink mmWave-NOMA system was formulated and solved in \cite{Zhu2018UplinkNOMA}. The joint power allocation/control and beamforming approaches in \cite{xiao2018mmWaveNOMA} and \cite{Zhu2018UplinkNOMA} can both achieve a near-upper-bound performance in terms of ASR. However, the schemes are confronted with great challenges to be generalized into a multiple user scenario.

Similar to \cite{xiao2018mmWaveNOMA} and \cite{Zhu2018UplinkNOMA}, we consider power allocation and analog beamforming for downlink mmWave-NOMA Networks in this paper. Particularly, the scenario is generalized to a multiple user case, and the power allocation is optimized jointly with Tx and Rx beamforming \footnote{Note that in a parallel submitted paper of us \cite{Xiao2018UserFairnessNOMA}, we also considered the generalization from a 2-user mmWave-NOMA to multi-user mmWave-NOMA, but the object in \cite{Xiao2018UserFairnessNOMA} is to maximize the minimal user rate via joint power allocation and Tx beamforming.}. We formulate an optimization problem to maximize the ASR of the multiple users, and meanwhile each user has a minimum rate constraint. To the best of our knowledge, this is the first work that considers \emph{joint Tx-Rx beamforming and power allocation} in mmWave-NOMA Networks. As the formulated problem is non-convex and it cannot be directly solved by using the existing optimization tools, we propose a sub-optimal solution with three stages. In the first stage, the optimal power allocation with a closed form is obtained for an arbitrary fixed Tx-Rx beamforming. In the second stage, we obtain the optimal Rx beamforming with a closed form for arbitrary fixed Tx beamforming. In the third stage, by substituting the optimal solutions of the previous two stages into the original problem, a Tx beamforming problem is formulated. We propose a boundary-compressed particle swarm optimization (BC-PSO) algorithm to solve this problem and obtain a sub-optimal solution. The simulation results show that the proposed solution can achieve a near-upper-bound performance in terms of ASR, which is significantly better than those of the state-of-the-art schemes and the conventional mmWave-OMA system.

The rest of the paper is organized as follows. In Section II, we present the system model and formulate the problem. In Section III, we propose the solution. In Section IV, simulation results are given to demonstrate the performance of the proposed solution, and the paper is concluded finally in Section V.

Symbol Notation: $a$, $\mathbf{a}$ and $\mathbf{A}$ denote a scalar variable and a vector, respectively. $(\cdot)^{\rm{T}}$ and $(\cdot)^{\rm{H}}$ denote transpose and conjugate transpose, respectively. $|\cdot|$ and $\|\cdot\|$ denote the absolute value and Euclidean norm, respectively. $\mathbb{E}(\cdot)$ denotes the expectation operation. $[\mathbf{a}]_i$ denotes the $i$-th entry of $\mathbf{a}$.

\section{System Model and Problem Formulation}
\subsection{System model}
\begin{figure}[t]
\begin{center}
  \includegraphics[width=8.8 cm]{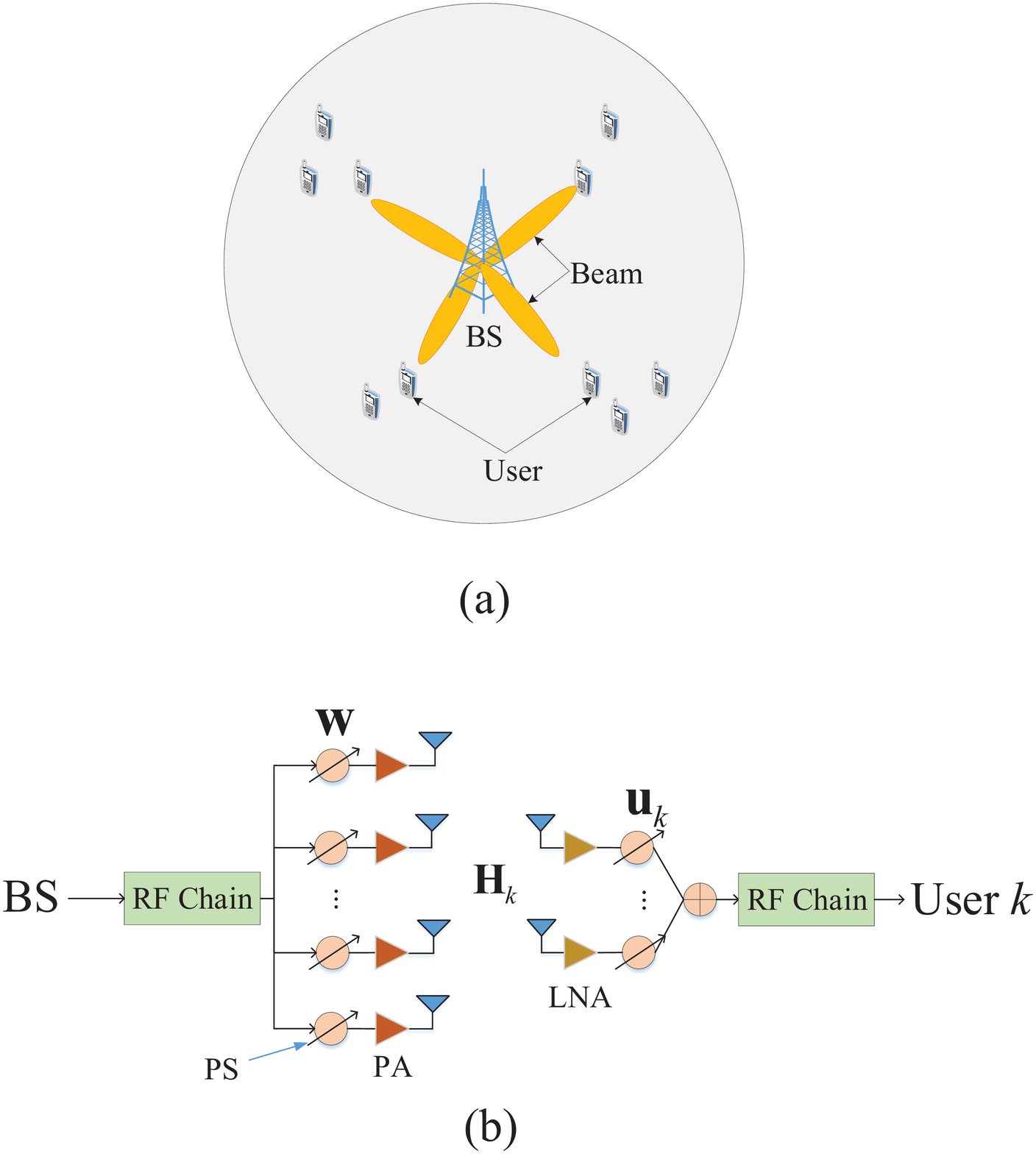}
  \caption{Illustration of a mmWave mobile cell, where the BS serves $K$ users simultaneously. The BS is equipped with a single RF chain and $N$ antennas, while each user is equipped with a single RF chain and $M$ antennas.}
  \label{fig:system}
\end{center}
\end{figure}
In this paper, we consider a downlink mmWave communications system. As shown in Fig. \ref{fig:system}, the BS serves $K$ users simultaneously. The numbers of the antennas equipped at the base station (BS) and each user are $N$ and $M$, respectively. Each antenna at the BS is driven by a phase shifter (PS) and a power amplifier (PA), while an antenna at the users is driven by a PS and low noise amplifier (LNA). Analog beamforming is utilized at both the BS and the user sides.

The BS transmits a signal $s_k$ to User $k~(k=1,2,\cdots, K)$ with transmission power $p_k$, where $\mathbb{E}(\left | s_{k} \right |^{2})=1$. The total transmission power of the BS is $P$. Thus, the received signal for User $k$ is
\begin{equation}
y_k=\mathbf{u}_{k}^{\rm{H}}\mathbf{H}_{k}\mathbf{w}\sum \limits_{k=1}^{K}\sqrt{p_{k}}s_{k}+n_{k},
\end{equation}
where $\mathbf{H}_{k}$ with dimension $M\times N$ is the channel response matrix between the BS and User $k$, and $n_k$ denotes the Gaussian white noise at User $k$ with power $\sigma^2$. $\mathbf{w}$ and $\mathbf{u}_{k}$ are the Tx beamforming vector (Tx-BFV) of the BS and the Rx beamforming vector (Rx-BFV) of User $k$, respectively. In general, the scaling factors of PA and LNA are constant. Thus, the Tx-BFV and Rx-BFV have constant modulus constraints, i.e.,
\begin{align}
&|[\mathbf{w}]_n|=\frac{1}{\sqrt{N}},~1\leq n \leq N,\\
&|[\mathbf{u}_{k}]_m|=\frac{1}{\sqrt{M}},~1\leq m \leq M,~1\leq k \leq K.
\end{align}

The channel between the BS and User $k$ is a mmWave channel \footnote{In this paper, we assume the channel is known by the BS. The mmWave channel estimation with low complexity can be referred in \cite{xiao2016codebook} and \cite{xiao2017codebook}.}. Subject to limited scattering in the mmWave band, multipath is mainly caused by reflection. As the number of the multipath components (MPCs) is small in general, the mmWave channel has directionality and appears spatial sparsity in the angle domain \cite{peng2015enhanced,wang2015multi,Lee2014exploiting,Gao2016ChannelEst,xiao2016codebook,alkhateeb2014channel}. Different MPCs have different angles of departure (AoDs) and angles of arrival (AoAs). Without loss of generality, we adopt the directional mmWave channel model assuming a uniform linear array (ULA) with a half-wavelength antenna space. Then, an mmWave channel between the BS and User $k$ can be expressed as \cite{peng2015enhanced,wang2015multi,Lee2014exploiting,Gao2016ChannelEst,xiao2016codebook,alkhateeb2014channel}
\begin{equation} \label{eq_oriChannel}
\mathbf{H}_{k}=\sum_{\ell=1}^{L_k}\lambda_{k,\ell}\mathbf{a}_{\rm{r}}(\theta_{k,\ell})\mathbf{a}_{\rm{t}}^{\rm{H}}(\psi_{k,\ell}),
\end{equation}
where $\lambda_{k,\ell}$, $\theta_{k,\ell}$ and $\psi_{k,\ell}$ are the complex coefficient, cos(AoD) and cos(AoA) of the $\ell$-th MPC of the channel vector for User $k$, respectively. $L_k$ is the total number of MPCs for User $k$, $\mathbf{a}_{\rm{t}}(\cdot)$ and $\mathbf{a}_{\rm{r}}(\cdot)$ are steering vectors defined as
\begin{align} \label{eq_steeringVCT}
&\mathbf{a}_{\rm{t}}(\theta)=[e^{j\pi0\theta},e^{j\pi1\theta},e^{j\pi2\theta},\cdots,e^{j\pi(N-1)\theta}]^{\mathrm{T}},\\
&\mathbf{a}_{\rm{r}}(\psi)=[e^{j\pi0\psi},e^{j\pi1\psi},e^{j\pi2\psi},\cdots,e^{j\pi(M-1)\psi}]^{\mathrm{T}},
\end{align}
which depend on the array geometry.

\subsection{Achievable Rate}
As discussed in the previous section, the optimal decoding order for NOMA is the increasing order of the users' channel gains in general. However, for the mmWave-NOMA with analog beamforming structure in this paper, the effective channel gains of the users are determined by both the channel gains and the beamforming gains. Thus, we need to sort the effective channel gains first, and then determine the decoding order. Without loss of generality, we assume that the order of the effective channel gains is $\left |\mathbf{u}_{\pi_1}^{\rm{H}}\mathbf{H}_{\pi_1}\mathbf{w} \right |^{2} \geq \left |\mathbf{u}_{\pi_2}^{\rm{H}}\mathbf{H}_{\pi_2}\mathbf{w} \right |^{2} \geq \cdots \geq \left |\mathbf{u}_{\pi_K}^{\rm{H}}\mathbf{H}_{\pi_K}\mathbf{w} \right |^{2}$, and thus the optimal decoding order is the increasing order of the effective channel gains \cite{xiao2018mmWaveNOMA,Ding2017random,Daill2017}. Therefore, User $\pi_k$ can decode $s_{\pi_n}~(k+1 \leq n \leq K)$ and then remove them from the received signal in a successive manner. The signals for User $\pi_{m}~(1\leq m \leq k-1)$ are treated as interference. Thus, the achievable rate of User $\pi_k$ is denoted by
\begin{equation}\label{eq_Rate}
R_{\pi_k}=\log_{2}(1+ \frac{\left | \mathbf{u}_{\pi_k}^{\rm{H}}\mathbf{H}_{\pi_k}\mathbf{w} \right |^{2}p_{\pi_k}}{\left | \mathbf{u}_{\pi_k}^{\rm{H}}\mathbf{H}_{\pi_k}\mathbf{w} \right |^{2}\sum \limits_{m=1}^{k-1}p_{\pi_m}+\sigma^{2}}).
\end{equation}

The achievable sum rate (ASR) of the proposed mmWave-NOMA system is
\begin{equation}
R_{\rm{sum}}=\sum\limits_{k=1}^{K}R_{k},
\end{equation}
where $R_k\in\{R_{\pi_k},k=1,2,...,K\}$ depending on the decoding order.

\subsection{Problem Formulation}
To improve the overall data rate, we formulate a joint Tx-Rx beamforming and power allocation problem to maximize the ASR of the $K$ users, where each user has a minimum rate constraint. The problem is formulated as
\begin{equation}\label{eq_problem}
\begin{aligned}
\mathop{\mathrm{Maximize}}\limits_{\{p_k\},\{\mathbf{u}_{k}\},\mathbf{w}}~~~~ &R_{\rm{sum}}\\
\mathrm{Subject~ to}~~~ &C_1~:~R_{k} \geq r_{k},~~\forall k, \\
&C_2~:~p_{k} \geq 0, ~~\forall k,\\
&C_3~:~\sum \limits_{k=1}^{K} p_{k} \leq P, \\
&C_4~:~|[\mathbf{u}_{k}]_{m}| = \frac{1}{\sqrt{M}},~~\forall k,m,\\
&C_5~:~|[\mathbf{w}]_{n}| = \frac{1}{\sqrt{N}},~~\forall n,
\end{aligned}
\end{equation}
where the constraint $C_1$ is the minimum rate constraint for each user. The constraint $C_2$ indicates that the power allocation to each user should be positive. The constraint $C_3$ is the total transmission power constraint, where the total power is no more than $P$. $C_4$ and $C_5$ are the CM constraints for the RxBFVs and TxBFV, respectively.

The total dimension of the variables in Problem \eqref{eq_problem} is $N+MK+K$, which is large in general. Direct search for the optimal solution results in heavy computational load, which is hard to accomplish in practice. To solve Problem \eqref{eq_problem}, there are two main challenges. One is that the optimized variables are entangled with each other, which makes the formulation non-convex. The other is that the expression of $R_{\rm{sum}}$ depends on the decoding order. In general, the optimal decoding order is the increasing order of the users' effective channel gains. However, the order of effective channel gains varies with different Tx-BFVs and Rx-BFVs. In other words, given different TxBFVs and RxBFVs, the objective function in Problem \eqref{eq_problem}, i.e., the ASR of the users, has different expressions. The two challenges make it infeasible to solve Problem \eqref{eq_problem} by using the existing optimization tools. Next, we will propose a sub-optimal solution with promising performance but low computational complexity.

\section{Solution of the Problem}
In this section, we propose a sub-optimal solution of Problem \eqref{eq_problem} with three stages. In the first stage, we obtain the optimal power allocation with a closed form for an arbitrary fixed Tx-BFV and arbitrary fixed Rx-BFVs. In the second stage, the optimal Rx-BFVs are obtained with a closed form for an arbitrary fixed Tx-BFV. In the third stage, we propose the BC-PSO algorithm to solve the Tx beamforming problem and obtain a sub-optimal Tx-BFV.
\subsection{Optimal Power Allocation with Arbitrary Fixed BFVs}
As we have analyzed before, an essential challenge to solve Problem \eqref{eq_problem} is the variation of the decoding order. However, given an arbitrary fixed Tx-BFV ${\bf{w}}$ and an arbitrary fixed Rx-BFVs ${\bf{u}}_k$, the order of the effective channel gains is fixed. For symbol simplicity and without loss of generality, we assume $\left |\mathbf{u}_{1}^{\rm{H}}\mathbf{H}_{1}\mathbf{w} \right |^{2} \geq \left |\mathbf{u}_{2}^{\rm{H}}\mathbf{H}_{2}\mathbf{w} \right |^{2} \geq \cdots \geq \left |\mathbf{u}_{K}^{\rm{H}}\mathbf{H}_{K}\mathbf{w} \right |^{2}$ in this subsection. The original problem can be simplified as
\begin{equation}\label{eq_problem2}
\begin{aligned}
\mathop{\mathrm{Maximize}}\limits_{\{p_k\}}~~~~ &R_{\rm{sum}}\\
\mathrm{Subject~ to}~~~ &C_1~:~R_{k} \geq r_{k},~~\forall k,\\
&C_2~:~p_{k} \geq 0, ~~\forall k,\\
&C_3~:~\sum \limits_{k=1}^{K} p_{k} \leq P.
\end{aligned}
\end{equation}
where the BFVs are arbitrary and fixed. To solve Problem \eqref{eq_problem2}, we give the following Lemma first.
\begin{lemma} The optimal power allocation in Problem \eqref{eq_problem2} must satisfy
\begin{equation}
\sum \limits_{k=1}^{K} p_{k} = P.
\end{equation}
\end{lemma}
\begin{proof}
We prove Lemma 1 by using contradiction. Denote the optimal power allocation in Problem \eqref{eq_problem2} as $\{p_{k}^{\star}\}$, and the achievable rate of User $k$ under optimal power allocation is denoted by $R_{k}^{\star}$. Assume $\sum \limits_{k=1}^{K} p_{k}^{\star}< P$.

Consider the following parameter settings
\begin{equation}\label{para_set}
\left\{
\begin{aligned}
&p_{k}=p_{k}^{\star}, ~1\leq k \leq K-1 \\
&p_{K}=P-\sum \limits_{k=1}^{K-1} p_{k}^{\star}.
\end{aligned}
\right.
\end{equation}
It is easy to verify that $R_{k}=R_{k}^{\star}~(1 \leq k \leq K-1)$ and $R_{K}>R_{K}^{\star}$, which means that the parameter settings in \eqref{para_set} can satisfy the minimum rate constraint as well as improve the ASR in Problem \eqref{eq_problem2}. It is in contrast to the assumption that $\{p_{k}^{\star}\}$ is optimal. Thus, we have $\sum \limits_{k=1}^{K} p_{k}^{\star}= P$.
\end{proof}

According to Lemma 1, Problem is equivalent to
\begin{equation}\label{eq_problem3}
\begin{aligned}
\mathop{\mathrm{Maximize}}\limits_{\{p_k\}}~~~~ &R_{\rm{sum}}\\
\mathrm{Subject~ to}~~~ &C_1~:~R_{k} \geq r_{k},~~\forall k,\\
&C_2~:~p_{k} \geq 0, ~~\forall k,\\
&C_3~:~\sum \limits_{k=1}^{K} p_{k} = P.
\end{aligned}
\end{equation}

As the number of users is $K$, it is difficult to directly obtain the optimal power allocation for all the users. Thus, we commence from a simplified case, where only two adjacent users can adjust the transmission power while the other users have fixed transmission power. The details are shown in the following Lemma.

\begin{lemma} For any $k_{0}$ ranging from 2 to $K$, if $p_{k}~ (k=1,2,\cdots,p_{k_{0}-2},p_{k_{0}+1},\cdots,p_{K})$ are all fixed, then $R_{\rm{sum}}$ in Problem \eqref{eq_problem3} is decreasing for $p_{k_{0}}$.
\end{lemma}
\begin{proof}
With fixed $p_{k}~ (k=1,2,\cdots,p_{k_{0}-2},p_{k_{0}+1},\cdots,p_{K})$, it is easy to verify that $R_{k}~ (k=1,2,\cdots,p_{k_{0}-2},p_{k_{0}+1},\cdots,p_{K})$ are constant. According to the constraint $C_{3}$ in Problem \eqref{eq_problem3}, we have
\begin{equation}
\begin{aligned}
&p_{k_{0}-1}+p_{k_{0}}+\sum\limits_{k\neq k_{0}-1,k_{0}}p_{k}=P\\
\Rightarrow &p_{k_{0}-1}=(P-\sum\limits_{k\neq k_{0}-1,k_{0}}p_{k})-p_{k_{0}}\triangleq \widetilde{P}-p_{k_{0}}.
\end{aligned}
\end{equation}
Thus, there is only one independent variable $p_{k_{0}}$ in Problem \eqref{eq_problem3}. The differential of the objective function $R_{\rm{sum}}$ is shown in \eqref{diff}.
\begin{figure*}
\begin{equation}\label{diff}
\begin{aligned}
&\frac{d~R_{\rm{sum}}}{d~p_{k_{0}}}=\frac{d~(R_{k_{0}}+R_{k_{0}-1})}{d~p_{k_{0}}}\\
=&\frac{d~[\log_{2}(1+\frac{\left |\mathbf{u}_{k_{0}}^{\rm{H}}\mathbf{H}_{k_{0}}\mathbf{w} \right |^{2}p_{k_{0}}}{\left |\mathbf{u}_{k_{0}}^{\rm{H}}\mathbf{H}_{k_{0}}\mathbf{w} \right |^{2}(\sum \limits_{m=1}^{k_{0}-2}p_{m}+\widetilde{P}-p_{k_{0}})+\sigma^{2}})]}{d~p_{k_{0}}}+\frac{d~[\log_{2}(1+\frac{\left |\mathbf{u}_{k_{0}-1}^{\rm{H}}\mathbf{H}_{k_{0}-1}\mathbf{w} \right |^{2}(\widetilde{P}-p_{k_{0}})}{\left |\mathbf{u}_{k_{0}-1}^{\rm{H}}\mathbf{H}_{k_{0}-1}\mathbf{w} \right |^{2}\sum \limits_{m=1}^{k_{0}-2}p_{m}+\sigma^{2}})]}{d~p_{k_{0}}}\\
=&\frac{1}{\ln2}\frac{\left |\mathbf{u}_{k_{0}}^{\rm{H}}\mathbf{H}_{k_{0}}\mathbf{w} \right |^{2}}{\left |\mathbf{u}_{k_{0}}^{\rm{H}}\mathbf{H}_{k_{0}}\mathbf{w} \right |^{2}\sum\limits_{m=1}^{k_{0}-1}p_{m}+\sigma^2}-\frac{1}{\ln2}\frac{\left |\mathbf{u}_{k_{0}-1}^{\rm{H}}\mathbf{H}_{k_{0}-1}\mathbf{w} \right |^{2}}{\left |\mathbf{u}_{k_{0}-1}^{\rm{H}}\mathbf{H}_{k_{0}-1}\mathbf{w} \right |^{2}\sum\limits_{m=1}^{k_{0}-1}p_{m}+\sigma^2}\\
=&\frac{1}{\ln2}\frac{(\left |\mathbf{u}_{k_{0}}^{\rm{H}}\mathbf{H}_{k_{0}}\mathbf{w} \right |^{2}-\left |\mathbf{u}_{k_{0}-1}^{\rm{H}}\mathbf{H}_{k_{0}-1}\mathbf{w} \right |^{2})\sigma^2}{(\left |\mathbf{u}_{k_{0}}^{\rm{H}}\mathbf{H}_{k_{0}}\mathbf{w} \right |^{2}\sum\limits_{m=1}^{k_{0}-1}p_{m}+\sigma^2)(\left |\mathbf{u}_{k_{0}-1}^{\rm{H}}\mathbf{H}_{k_{0}-1}\mathbf{w} \right |^{2}\sum\limits_{m=1}^{k_{0}-1}p_{m}+\sigma^2)}.
\end{aligned}
\end{equation}
\hrulefill
\end{figure*}

As we have assumed that $\left |\mathbf{u}_{1}^{\rm{H}}\mathbf{H}_{1}\mathbf{w} \right |^{2} \geq \left |\mathbf{u}_{2}^{\rm{H}}\mathbf{H}_{2}\mathbf{w} \right |^{2} \geq \cdots \geq \left |\mathbf{u}_{K}^{\rm{H}}\mathbf{H}_{K}\mathbf{w} \right |^{2}$, thus $\frac{d R_{\rm{sum}}}{d p_{k_{0}}}\leq 0$. We can conclude that $R_{\rm{sum}}$ is decreasing for $p_{k_{0}}$.
\end{proof}

Based on Lemma 2, we can find that the priority of power allocation in Problem \eqref{eq_problem3} is $p_{1}\succ p_{2} \succ \cdots \succ p_{K}$, where $\succ$ denotes higher priority. In other words, the BS tends to allocate more power to the user with the highest effective channel gain, and meanwhile allocate only necessary power to the other users to satisfy the minimum rate constraints. We give the following Theorem to illustrate this property and obtain the optimal power allocation.
\begin{theorem} The optimal solution in Problem \eqref{eq_problem3} must satisfy $R_{k}=r_{k}~(2\leq k\leq K)$, and the optimal power allocation is given by
\begin{equation}\label{opt_power}
\left \{
\begin{aligned}
&p_{K}^{\star}=\frac{\eta_{K}}{\eta_{K}+1}(P+\frac{\sigma^2}{\left |\mathbf{u}_{K}^{\rm{H}}\mathbf{H}_{K}\mathbf{w} \right |^{2}}),\\
&p_{K-1}^{\star}=\frac{\eta_{K-1}}{\eta_{K-1}+1}(P-p_{K}^{\star}+\frac{\sigma^2}{\left |\mathbf{u}_{K-1}^{\rm{H}}\mathbf{H}_{K-1}\mathbf{w} \right |^{2}}),\\
&~~~~\vdots\\
&p_{2}^{\star}=\frac{\eta_{2}}{\eta_{2}+1}(P-\sum\limits_{m=3}^{K}p_{m}^{\star}+\frac{\sigma^2}{\left |\mathbf{u}_{2}^{\rm{H}}\mathbf{H}_{2}\mathbf{w} \right |^{2}}),\\
&p_{1}^{\star}=P-\sum\limits_{m=2}^{K}p_{m}^{\star},
\end{aligned}
\right.
\end{equation}
where $\eta_{k}=2^{r_{k}}-1$.
\end{theorem}
\begin{proof}
See Appendix A.
\end{proof}

Based on Theorem 1, we can find that although the users with lower effective channel gains are prior in the decoding order, the achievable rates of them have no gain compared with the minimum rate constraints. The power allocated to Users $2$-$K$ is only necessary to satisfy the minimum rate constraint. The performance gain of mmWave-NOMA depends mainly on the user with the highest effective channel gain, i.e., User 1.

\subsection{Optimal Rx-BFVs with an arbitrary fixed Tx-BFV}
In the first stage, we obtained the closed-form power allocation with arbitrary fixed BFVs as shown in \eqref{opt_power}. In the second stage, we will handle the Rx beamforming. Given an arbitrary fixed Tx-BFV, Problem \eqref{eq_problem} is simplified as
\begin{equation}\label{eq_problem4}
\begin{aligned}
\mathop{\mathrm{Maximize}}\limits_{\{p_k\},\{\mathbf{u}_{k}\}}~~~~ &R_{\rm{sum}}\\
\mathrm{Subject~ to}~~~ &C_1~:~R_{k} \geq r_{k},~~\forall k, \\
&C_2~:~p_{k} \geq 0, ~~\forall k,\\
&C_3~:~\sum \limits_{k=1}^{K} p_{k} \leq P, \\
&C_4~:~|[\mathbf{u}_{k}]_{m}| = \frac{1}{\sqrt{M}},~~\forall k,m.
\end{aligned}
\end{equation}

To obtain the optimal Rx beamforming, we have the following Theorem.
\begin{theorem} The optimal solution of the Rx-BFVs in Problem \eqref{eq_problem4} is
\begin{equation}\label{opt_RxBFV}
[\mathbf{u}_{k}^{\star}]_{m}=\frac{1}{\sqrt{M}}\frac{[\mathbf{H}_{k}\mathbf{w}]_{m}}{|[\mathbf{H}_{k}\mathbf{w}]_{m}|},~~\forall k,m.
\end{equation}
\end{theorem}

\begin{proof}
As the Tx-BFV is fixed, $\mathbf{H}_{k}\mathbf{w}~(1\leq k \leq K)$ are all constant vectors. Given an arbitrary decoding order of $\pi_{K},\pi_{K-1},\cdots,\pi_{1}$, we introduce intermediate variables $b_{\pi_k}$, where $b_{\pi_k}=\left | \mathbf{u}_{\pi_k}^{\rm{H}}\mathbf{H}_{\pi_k}\mathbf{w} \right |^{2}$ for $1\leq k\leq K$. Thus, the partial derivative of the achievable rate is

\begin{align}
&\notag \frac{\partial~R_{\pi_{s}}}{\partial~b_{\pi_{k}}}\mid \{\pi_{s}=\pi_{k}\}=\frac{\partial~\log_{2}(1+ \frac{b_{\pi_{k}}p_{\pi_{k}}}{b_{\pi_{k}}\sum \limits_{m=1}^{k-1}p_{\pi_{m}}+\sigma^{2}})}{\partial~b_{\pi_{k}}}\\
&~~~~~=\frac{1}{\ln2}\frac{p_{\pi_{k}}\sigma^2}{(b_{\pi_{k}}\sum \limits_{m=1}^{k-1}p_{\pi_{m}}+\sigma^{2})(b_{\pi_{k}}\sum \limits_{m=1}^{k}p_{\pi_{m}}+\sigma^{2})}\geq 0,\\
&\frac{\partial~R_{\pi_{s}}}{\partial~b_{\pi_{k}}}\mid \{\pi_{s}\neq \pi_{k}\}=0.
\end{align}

The achievable rate of User $\pi_{k}$ is increasing for $b_{\pi_{k}}$, while the achievable rates of the other users are independent to $b_{\pi_{k}}$. Thus, to maximize the ASR, we can always adjust the Rx-BFV for each user to maximize $b_{\pi_{k}}~(1\leq k\leq K)$. For User $\pi_{k}$, as $\mathbf{H}_{\pi_{k}}\mathbf{w}$ is a constant vector, we just need to let the phase of each element of $[\mathbf{u}_{\pi_{k}}]$ be the same with the phase of the corresponding element of $[\mathbf{H}_{\pi_{k}}\mathbf{w}]$, which is not influenced by the decoding order. Thus, under any decoding orders, the optimal solution of the Rx-BFVs is always given by \eqref{opt_RxBFV}.
\end{proof}

Based on Theorem 1 and Theorem 2, we can further obtain the ASR of the $K$ users, which is given by
\begin{equation}\label{SumRate}
\begin{aligned}
&R_{\rm{sum}} \triangleq R(\mathbf{w})\\
=&\sum\limits_{k=2}^{K}r_{k}+\log_{2}(1+ \frac{\left | \mathbf{u}_{1}^{\star\rm{H}}\mathbf{H}_{1}\mathbf{w} \right |^{2}p_{1}^{\star}}{\sigma^{2}}),
\end{aligned}
\end{equation}
where $p_{1}^{\star}$ and $\mathbf{u}_{1}^{\star}$ are both functions of $\mathbf{w}$, whose definitions are given by \eqref{opt_power} and \eqref{opt_RxBFV}, respectively. From \eqref{SumRate}, we can find the value of $R_{\rm{sum}}$ is only determined by the Tx-BFV. Next, we will give the approach of Tx beamforming design in the third stage.

\subsection{Design of Tx-BFV with BC-PSO}
According to Theorem 1 and Theorem 2, Problem \eqref{eq_problem} can be transformed into a Tx beamforming problem, i.e.,
\begin{equation}\label{eq_BFproblem}
\begin{aligned}
\mathop{\mathrm{Maximize}}\limits_{\bf{w}}~~~~ &R(\mathbf{w})\\
\mathrm{Subject~ to}~~~ &|[\mathbf{w}]_{n}| = \frac{1}{\sqrt{N}},~~1\leq n \leq N
\end{aligned}
\end{equation}
where $R(\mathbf{w})$ is the achievable sum rate of $K$ users shown in \eqref{SumRate}. Although the explicit expression of $R(\mathbf{w})$ can be obtained according to \eqref{opt_power}, \eqref{opt_RxBFV} and \eqref{SumRate}, the highly non-convex formulation makes it complicated to solve Problem \eqref{eq_BFproblem} directly. In addition, the dimension of the Tx-BFV, i.e., $N$, is large in general, so it is computationally prohibitive to direct search the optimal solution.

To solve this difficult problem, particle swarm optimization (PSO) is an alternative approach \cite{Rahmat-Samii2003PSO,Robinson2004PSO,fukuyama2008fundamentals}. First, we give the basics of PSO.
\subsubsection{Basics of PSO }
In the $N$-dimension search space $\mathcal{S}$, the $I$ particles in the swarm are randomly initialized with position $\mathbf{x}$ and velocity $\mathbf{v}$. Each particle has a memory for its best found position $\mathbf{p}_{\text{best}}$ and the globally best position $\mathbf{g}_{\text{best}}$, where the goodness of a position is evaluated by the fitness function. For each iteration, the velocity and position of each particle are updated based on
\begin{equation}\label{eq_PSOregular}
\begin{aligned}
&[\mathbf{v}]_{n}=\omega[\mathbf{v}]_{n}+c_{1}\text{rand()}*([\mathbf{p}_{\text{best}}]_{n}-[\mathbf{x}]_{n})\\
&~~~~~~~~~~~~~~~~+c_{2}\text{rand()}*([\mathbf{g}_{\text{best}}]_{n}-[\mathbf{x}]_{n})\\
&[\mathbf{x}]_{n}=[\mathbf{x}]_{n}+[\mathbf{v}]_{n}
\end{aligned}
\end{equation}
for $n=1,2,\cdots,N$. The parameter $\omega$ is the inertia weight of velocity. In general, $\omega$ is linearly decreasing to improve the convergence speed. The parameters $c_{1}$ and $c_{2}$ are the cognitive ratio and social ratio, respectively. The random number function rand() returns a number between 0.0 and 1.0 with uniform distribution. After calculating the fitness function for each particle, the locally and globally best positions, i.e., $\mathbf{p}_{\text{best}}$ and $\mathbf{g}_{\text{best}}$, are updated. In such a manner, the particles diffuse around the search space and may find the globally optimal solution.

However, the CM constraint in Problem \eqref{eq_BFproblem} makes the search space highly non-convex. The particles may converge to a locally optimal solution with a high probability. Thus, directly using PSO in Problem \eqref{eq_BFproblem} may not obtain a promising performance. To this end, we propose a modified approach, i.e., BC-PSO. In the proposed approach, the feasible region is relaxed to a convex set, i.e., $|[\mathbf{w}]_{n}| \leq \frac{1}{\sqrt{N}}$. The boundary-compressed approach is proposed to guarantee that the particles satisfy the CM constraint. The details of the BC-PSO algorithm is shown bellow.

\subsubsection{Implementation of BC-PSO}
Define the search space of Problem \eqref{eq_BFproblem} as $\mathcal{S}=\{\mathbf{w}\big{|}|[\mathbf{w}]_{n}| \leq \frac{1}{\sqrt{N}}, 1\leq n \leq N\}$, which has two boundaries. The outer boundary is defined as $\{|[\mathbf{w}]_{n}| = \frac{1}{\sqrt{N}}, 1\leq n \leq N\}$, while the inner boundary is $\{|[\mathbf{w}]_{n}| = d_{t}, 1\leq n \leq N\}$. $d_{t}$ is a dynamic parameter, which is linear to the number of iterations. The initial value of $d_{t}$ is 0, and it increases linearly for each iteration until $d_{t}=\frac{1}{\sqrt{N}}$. For each iteration, if the particle moves across the outer/inner boundary, then it is adjusted onto the boundary. With this implementation, the particles can move throughout the relaxed search space and converges to the outer boundary eventually.

On the other hand, the definitions of the fitness function for different particles are different. The reason is that the order of effective channel gains may change when the particles move, which results in the change of the ASR's expression . Thus, when implementing the BC-PSO algorithm here, we should reorder the effective channel gains first in each iteration, and then obtain the fitness function, i.e., $R(\mathbf{w})$, according to \eqref{SumRate}.

In summary, we give Algorithm \ref{alg1} to solve Problem \eqref{eq_problem}.

\begin{algorithm}[h]
\caption{Implementation of BC-PSO}
\label{alg1}
\begin{algorithmic}[1]
\REQUIRE ~\\
Number of antennas: $M$ and $N$;\\
Number of particle swarm: $I$;\\
Maximum number of iterations: $T$;\\
Scaling factors: $c_1$ and $c_2$;\\
Range of inertia weight: $\omega_{\text{max}}$ and $\omega_{\text{min}}$.
\ENSURE ~$p_{k}^{\star}$, $\mathbf{u}_{k}^{\star}$ and $\mathbf{w}^{\star}$\\
\STATE Initialize the position $\mathbf{x}_{i}=\mathbf{w}_{i}$ and velocity $\mathbf{v}_{i}$.
\STATE Find the globally best solution position $\mathbf{g}_{\text{best}}$.
\FOR {$t=1:T$ }
\STATE $\omega=\omega_{\text{max}}-\frac{t}{T}(\omega_{\text{max}}-\omega_{\text{min}})$.
\STATE $d_{t}=\frac{t}{T\sqrt{N}}$.
\FOR {$i=1:I$ }
\FOR {$n=1:N$ }
\STATE Update $[\mathbf{v}_{i}]_{n}$ and $[\mathbf{x}_{i}]_{n}$ based on \eqref{eq_PSOregular}.
\IF{$|[\mathbf{x}_{i}]_{n}|<d_{t}$}
\STATE $[\mathbf{x}_{i}]_{n}=\frac{d_{t}[\mathbf{x}_{i}]_{n}}{|[\mathbf{x}_{i}]_{n}|}$.
\ENDIF
\IF {$|[\mathbf{x}_{i}]_{n}|>\frac{1}{\sqrt{N}}$}
\STATE $[\mathbf{x}_{i}]_{n}=\frac{[\mathbf{x}_{i}]_{n}}{\sqrt{N}|[\mathbf{x}_{i}]_{n}|}$.
\ENDIF
\IF{$|[\mathbf{p}_{\text{best},i}]_{n}|<d_{t}$}
\STATE $[\mathbf{p}_{\text{best},i}]_{n}=\frac{d_{t}[\mathbf{p}_{\text{best},i}]_{n}}{|[\mathbf{p}_{\text{best},i}]_{n}|}$.
\ENDIF
\STATE Obtain the optimal Rx-BFVs $\mathbf{u}_{k}^{\star}$ according to \eqref{opt_RxBFV}.
\STATE Reorder the effective channel gains of the users.
\STATE Obtain the optimal power allocation $p_{k}^{\star}$ according to \eqref{opt_power}.
\STATE Obtain the fitness function $R(\mathbf{w})$ according to \eqref{SumRate}.
\ENDFOR
\STATE Update $\mathbf{p}_{\text{best},i}$.
\ENDFOR
\STATE Update $\mathbf{g}_{\text{best}}$.
\ENDFOR
\STATE $\mathbf{w}^{\star}=\mathbf{g}_{\text{best}}$.
\RETURN $p_{k}^{\star}$, $\mathbf{u}_{k}^{\star}$ and $\mathbf{w}^{\star}$.
\end{algorithmic}
\end{algorithm}

Hereto, we solve the original problem. In the proposed solution, the power allocation and Rx beamforming are optimal, while the Tx beamforming is sub-optimal.

% The implementation of BC-PSO solve the joint Tx/Rx beamforming and power allocation problem in the proposed mmWave-NOMA system. It can also be generalized to other analog beamforming problems. After defining the fitness function, BC-PSO can obtain a promising solution and meanwhile guarantee the CM constraint.

\subsection{Computational Complexity}
As we obtained the closed-form optimal power allocation and Rx-BFVs with an arbitrary fixed Tx-BFV, the computational complexity is mainly caused by Tx beamforming design in the third stage. In Algorithm 1, the total computational complexity is $\mathcal{O}(N)$, which linearly increases as $N$ and does not increase as $M$ or $K$. In contrast, if the direct search method is adopted, and the number of the candidate values for each variable in Problem \eqref{eq_problem} is $G$, the complexity of directly searching the globally optimal solution is $\mathcal{O}(G^{N+MK+K})$, which exponentially increases as $N$, $M$ and $K$.

\section{Performance Simulations}
In this section, we provide the simulation results to verify the performance of the proposed joint Tx-Rx beamforming and power allocation approach in the mmWave-NOMA system. We adopt the channel model shown in \eqref{eq_oriChannel}, where the users are uniformly distributed from 10m to 500m away from the BS, and the channel gain of the node 100m away from the BS has an average power of 0 dB to noise power. The number of MPCs for all the users is $L=4$. Both LOS and NLOS channel models are considered. For the LOS channel, the average power of the NLOS paths is 15 dB weaker than that of the LOS path. For the NLOS channel, the coefficient of each path has an average power of $1/\sqrt{L}$. For each channel realization in the simulations, the channel gains of the users are sorted by $\|\mathbf{h}_{1}\|\geq \|\mathbf{h}_{2}\|\geq \cdots \geq \|\mathbf{h}_{K}\|$. The corresponding parameter settings in Algorithm 1 are $I=800,T=50,c_{1}=c_{2}=1.4$.

\begin{figure}[t]
\begin{center}
  \includegraphics[width=\figwidth cm]{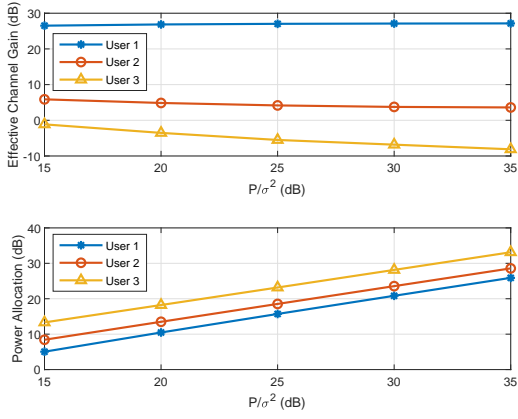}
  \caption{The values of effective channel gains and power allocation for different users with varying total power to noise ratio, where $M=4$, $N=16$, $K=3$ and $r_{k}=1.5$ bps/Hz.}
  \label{fig:Parameter_P}
\end{center}
\end{figure}
\begin{figure}[t]
\begin{center}
  \includegraphics[width=\figwidth cm]{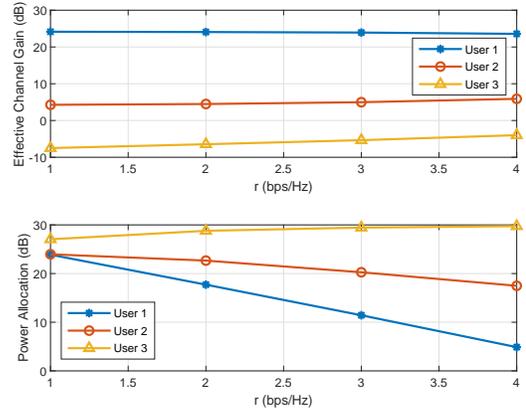}
  \caption{The values of effective channel gains and power allocation for different users with varying minimum rate constraint, where $M=4$, $N=16$, $K=3$, $r_{k}=r$ and $P/\sigma^2=30$ dB.}
  \label{fig:Parameter_r}
\end{center}
\end{figure}
We first show the power allocation and the effective channel gains in Figs. \ref{fig:Parameter_P} and \ref{fig:Parameter_r}. Each point in Figs. \ref{fig:Parameter_P} and \ref{fig:Parameter_r} is an average result of $10^3$ LOS channel realizations. From the two figures, we can find that the effective channel gain of User 1, the user with the highest channel gain, is distinctly larger than that of the other users. The user with a better channel gain has a higher effective channel gain with the proposed solution. In Fig. \ref{fig:Parameter_P}, the effective channel gains of User 1 and User 3 go increasing and decreasing, respectively, when $P/\sigma^2$ becomes higher. It indicates that when the total transmission power is high, power and beam gains should be allocated jointly to enlarge the difference of the effective channel gains to obtain a higher ASR. In contrast, the power allocation and the effective channel gain of User 1 go decreasing while the power allocation and the effective channel gain of User 3 go increasing, when $r$ becomes higher in Fig. \ref{fig:Parameter_r}. It indicates that more power and beam gain should be allocated to the users with worse channel gains to satisfy the constraint, when the minimum rate constraint is high.

\begin{figure}[t]
\begin{center}
  \includegraphics[width=\figwidth cm]{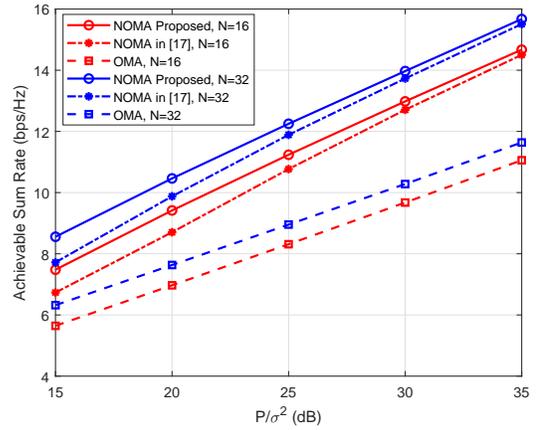}
  \caption{Comparison of the ASRs between the NOMA and OMA systems with varying total power to noise ratio, where $M=1$, $N=16,32$, $K=2$ and $r_{k}=1.5$ bps/Hz.}
  \label{fig:RateComp_P}
\end{center}
\end{figure}

Next, we compare the performance between the considered mmWave-NOMA system and a mmWave-OMA system. The achievable rate of User $k$ in a mmWave-OMA system is
\begin{equation}\label{OMA_Rate}
R_{k}^{\mathrm{(OMA)}}=\frac{1}{K}\log_{2}(1+ \frac{\left | \mathbf{u}_{k}^{\star\rm{H}}\mathbf{H}_{k}\mathbf{w}^{\star} \right |^{2}P}{\sigma^{2}}),
\end{equation}
where the factor $\frac{1}{K}$ is due to the multiplexing loss in OMA. $\mathbf{u}_{k}^{\star}$ and $\mathbf{w}^{\star}$ are the Rx-BFV and Tx-BFV given in Algorithm 1, respectively.

Fig. \ref{fig:RateComp_P} compares the ASRs between the proposed mmWave-NOMA algorithm, the mmWave-NOMA approach in \cite{xiao2018mmWaveNOMA} and mmWave-OMA with varying total power to noise ratio. Each point in Fig. \ref{fig:RateComp_P} is the average performance of $10^3$ LOS channel realizations. Significantly, the performance of the proposed mmWave-NOMA system is distinctly better than that of the mmWave-OMA system, as well as better than that of the solution in \cite{xiao2018mmWaveNOMA}. Particularly when $P/\sigma^2$ is low, the superiority of the proposed algorithm is more conspicuous compared with the approach in \cite{xiao2018mmWaveNOMA}. The reason is that given a designed BFV, the solutions of power allocation in this paper and \cite{xiao2018mmWaveNOMA} are both optimal. Thus, the performance gap is mainly caused by the beamforming design. Significantly, the proposed algorithm can always find a better BFV than that of the approach in \cite{xiao2018mmWaveNOMA}. As shown in \eqref{eq_Rate}, the achievable rate is determined by the product of the effective channel gain and the transmission power. When the total transmission power becomes lower, the effective channel gain becomes the main portion to determine the ASR, so the superiority of the proposed mmWave-NOMA algorithm is relatively conspicuous.

\begin{figure}[t]
\begin{center}
  \includegraphics[width=\figwidth cm]{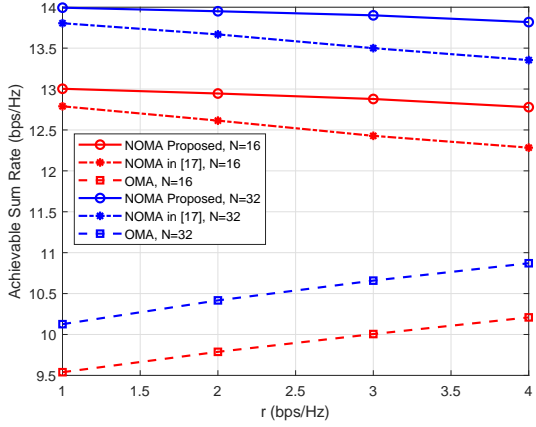}
  \caption{Comparison of the ASRs between the NOMA and OMA systems with varying minimum rate constraint, where $M=1$, $N=16,32$, $K=2$, $r_{k}=r$ and $P/\sigma^2=30$ dB.}
  \label{fig:RateComp_r}
\end{center}
\end{figure}
Fig. \ref{fig:RateComp_r} shows the comparison result of the ASRs between the proposed mmWave-NOMA algorithm, the mmWave-NOMA approach in \cite{xiao2018mmWaveNOMA} and mmWave-OMA with varying minimum rate constraint. Each point in Fig. \ref{fig:RateComp_r} is the average performance of $10^3$ LOS channel realizations. Similar to the result in Fig. \ref{fig:RateComp_P}, we can find that the proposed mmWave-NOMA algorithm can achieve a higher ASR than that of mmWave-NOMA in \cite{xiao2018mmWaveNOMA}, as well as higher than the ASR of the mmWave-OMA system. Particularly when $r$ increases, the superiority of the proposed algorithm is more conspicuous compared with the approach in \cite{xiao2018mmWaveNOMA}. The results indicate that the proposed beamforming design is better than that of the approach in \cite{xiao2018mmWaveNOMA}, especially when the minimum rate constraint is large.

\begin{figure}[t]
\begin{center}
  \includegraphics[width=\figwidth cm]{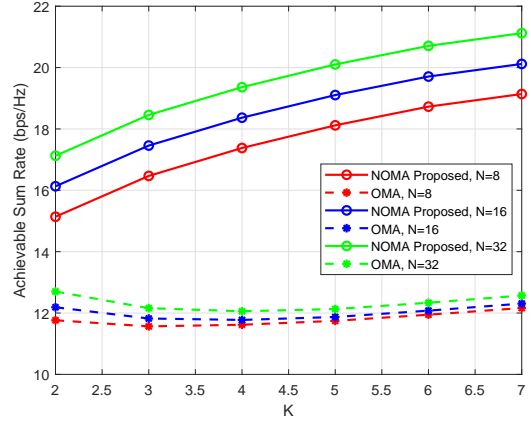}
  \caption{Comparison of the ASRs between the NOMA and OMA systems with varying number of the users, where $M=4$, $N=8,16,32$, $r_{k}=1.5$ bps/Hz and $\frac{P}{K\sigma^2}=30$ dB.}
  \label{fig:RateComp_K}
\end{center}
\end{figure}

Fig. \ref{fig:RateComp_K} compares the ASRs between mmWave-NOMA and mmWave-OMA systems with varying number of users. For fairness, each user has an average transmission power to noise of 30 dB. Each point in Fig. \ref{fig:RateComp_K} is the average performance of $10^3$ LOS channel realizations. It can be observed again that the mmWave-NOMA can outperform the mmWave-OMA, especially when $N$ is large. It can be seen that the ASR of the NOMA users improves as the number of users increases, while the ASR of the OMA users is always around a low value without obvious improvement. The results prove that mmWave-NOMA can achieve a higher spectrum efficiency compared with mmWave-OMA when the number of users increases.

\begin{figure}[t]
\begin{center}
  \includegraphics[width=\figwidth cm]{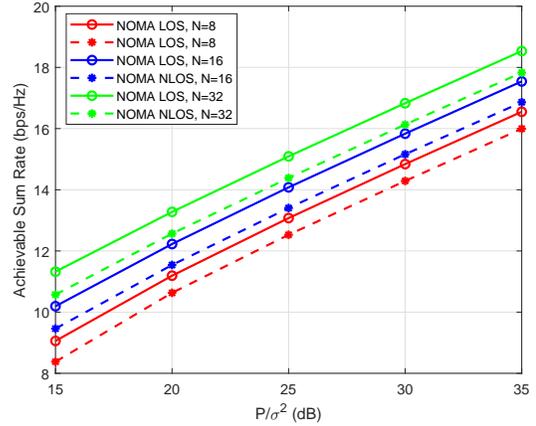}
  \caption{Comparison of the ASRs between the LOS and NLOS channel models with varying total power to noise ratio, where $M=4$, $N=8,16,32$, $K=3$ and $r_{k}=1.5$ bps/Hz.}
  \label{fig:Rate_LOS_NLOS_P}
\end{center}
\end{figure}

\begin{figure}[t]
\begin{center}
  \includegraphics[width=\figwidth cm]{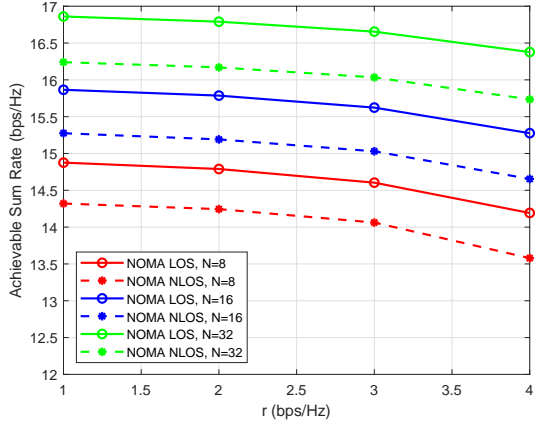}
  \caption{Comparison of the ASRs between the LOS and NLOS channel models with varying minimum rate constraint, where $M=4$, $N=8,16,32$, $K=3$, $r_{k}=r$ and $P/\sigma^2=30$ dB.}
  \label{fig:Rate_LOS_NLOS_r}
\end{center}
\end{figure}

Figs. \ref{fig:Rate_LOS_NLOS_P} and \ref{fig:Rate_LOS_NLOS_r} compare the ASRs of mmWave-NOMA system between the LOS and NLOS channel models with varying total power to noise ratio and with varying minimum rate constraint, respectively. Each point in this two figures is the average performance of $10^3$ channel realizations. It can be seen that the performance with the LOS channel model is distinctly better than that with the NLOS channel model, because the beam gain is more centralized for the LOS channel. Particularly, when$P/\sigma^2$ is small and $r$ is large, the performance gap between the LOS channel model and NLOS channel model is larger.

\begin{figure}[t]
\begin{center}
  \includegraphics[width=\figwidth cm]{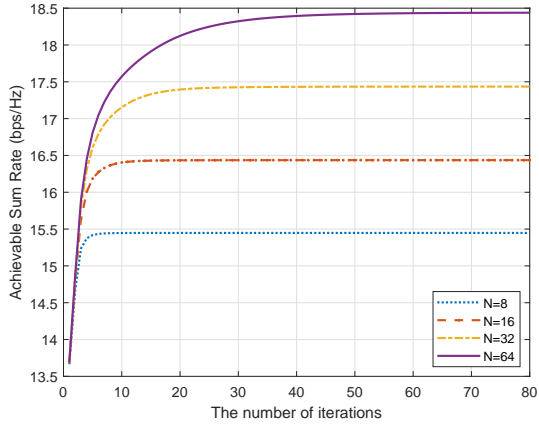}
  \caption{Iterations required for convergence in the BC-PSO algorithm, where $M=4$, $N=8,16,32,64$, $K=4$, $r_{k}=1$ bps/Hz and $P/\sigma^2=30$ dB.}
  \label{fig:Convergence_N}
\end{center}
\end{figure}

In the third stage of the solution, we proposed the BC-PSO algorithm and obtained a sub-optimal solution. The convergence of the proposed algorithm is evaluated in Fig. \ref{fig:Convergence_N}. When $N=8,16,32,64$, the curve of the ASR tends to be stable after 7,15,30,60 iterations, respectively. We can find that the number of iterations that the algorithm converges is roughly linear to $N$, which indicates that the proposed BC-PSO algorithm has a linear convergence rate against the number of antennas at the BS.

\begin{figure}[t]
\begin{center}
  \includegraphics[width=\figwidth cm]{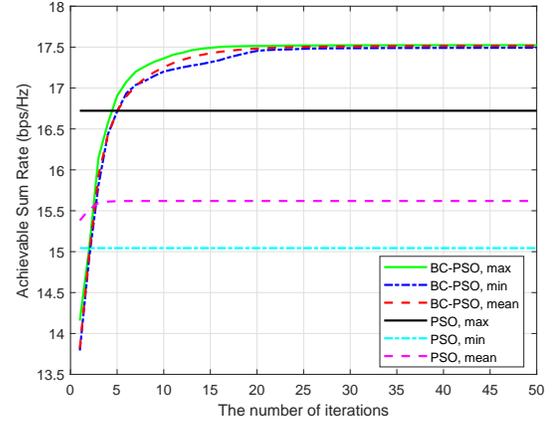}
  \caption{Comparison of the performance of PSO and BC-PSO, where $M=4$, $N=32$, $K=4$, $r_{k}=1$ bps/Hz and $P/\sigma^2=30$ dB.}
  \label{fig:Bound}
\end{center}
\end{figure}

To evaluate the stability of the proposed approach, we compare the performance of the proposed BC-PSO and the classical PSO in Fig. \ref{fig:Bound}, where PSO is corresponding to directly solving Problem \eqref{eq_BFproblem} in the search space of $\mathcal{S}=\{\mathbf{w}\big{|}|[\mathbf{w}]_{n}| = \frac{1}{\sqrt{N}}, 1\leq n \leq N\}$, while BC-PSO is corresponding to the proposed approach in Algorithm 1. With the same one channel realization, we solve Problem \eqref{eq_BFproblem} with the PSO algorithm and the BC-PSO algorithm for 1000 times with different initializations. It can be seen that the ASR with BC-PSO is distinctly higher than that with PSO. Moreover, the curve of the maximal ASR, minimal ASR and mean ASR for BC-PSO are close, which outperforms the performance of PSO. The results prove that the proposed BC-PSO algorithm can converge to a near-upper-bound performance with a high probability.

\section{Conclusion}
In this paper, we formulated a joint Tx-Rx beamforming and power allocation problem in a mmWave-NOMA system to maximize the ASR of the multiple users, subject to a minimum rate constraint for each user. To solve this non-convex problem, we proposed a sub-optimal solution with three stages. In the first stage, the optimal power allocation with a closed form was obtained for arbitrary fixed Tx-Rx beamforming. In the second stage, the optimal Rx beamforming with a closed form was designed for arbitrary fixed Tx beamforming. In the third stage, we proposed the BC-PSO algorithm to implement the reduced problem, i.e., a Tx beamforming problem, and obtained a sub-optimal solution. It was shown that the proposed algorithm has a preferable convergence and stability. The results showed that by using the proposed approach, the ASR of the mmWave-NOMA system could achieve a near-upper-bound performance with a high probability, which is significantly improved compared with the state-of-the-art approach and the conventional mmWave-OMA system.

\appendices
\section{Proof of Theorem1}
Denote the optimal power allocation of Problem \eqref{eq_problem3} is $\{p_{k}^{\star}\}$, and the achievable rate of User $k$ under optimal power allocation is $R_{k}^{\star}$. Assume that there is one user whose achievable rate is lager than its minimum rate constraint, i.e., $R_{k_{0}}^{\star} > r_{k_{0}}$, where $k_{0}$ is ranging from 2 to $K$. Consider the parameter settings bellow,
\begin{equation}\label{para_set2}
\left\{
\begin{aligned}
&p_{k}=p_{k}^{\star},~(k=1,2,\cdots,k_{0}-2,k_{0}+1,\cdots,K) \\
&p_{k_{0}-1}=p_{k_{0}-1}^{\star}+\delta, \\
&p_{k_{0}}=p_{k_{0}}^{\star}-\delta,
\end{aligned}
\right.
\end{equation}
where $\delta=\frac{S+\left | \mathbf{u}_{k_{0}}^{\rm{H}}\mathbf{H}_{k_{0}}\mathbf{w} \right |^{2}p_{k_{0}}^{\star}-\sqrt{2^{r_{k_{0}}}S(S+\left | \mathbf{u}_{k_{0}}^{\rm{H}}\mathbf{H}_{k_{0}}\mathbf{w} \right |^{2}p_{k_{0}}^{\star})}}{\left | \mathbf{u}_{k_{0}}^{\rm{H}}\mathbf{H}_{k_{0}}\mathbf{w} \right |^{2}}$ and $S=\left | \mathbf{u}_{k_{0}}^{\rm{H}}\mathbf{H}_{k_{0}}\mathbf{w} \right |^{2}\sum \limits_{m=1}^{k_{0}-1}p_{m}^{\star}+\sigma^{2}$. According to the assumption of $R_{k_{0}}^{\star} > r_{k_{0}}$, we have
\begin{equation}
\begin{aligned}
&1+ \frac{\left | \mathbf{u}_{k_{0}}^{\rm{H}}\mathbf{H}_{k_{0}}\mathbf{w} \right |^{2}p_{k_{0}}^{\star}}{\left | \mathbf{u}_{k_{0}}^{\rm{H}}\mathbf{H}_{k_{0}}\mathbf{w} \right |^{2}\sum \limits_{m=1}^{k_{0}-1}p_{m}^{\star}+\sigma^{2}}>2^{r_{k_{0}}}\\
\Leftrightarrow&S+\left | \mathbf{u}_{k_{0}}^{\rm{H}}\mathbf{H}_{k_{0}}\mathbf{w} \right |^{2}p_{k_{0}}^{\star} > 2^{r_{k_{0}}}S\\
\Leftrightarrow&(S+\left | \mathbf{u}_{k_{0}}^{\rm{H}}\mathbf{H}_{k_{0}}\mathbf{w} \right |^{2}p_{k_{0}}^{\star})^{2} > 2^{r_{k_{0}}}S(S+\left | \mathbf{u}_{k_{0}}^{\rm{H}}\mathbf{H}_{k_{0}}\mathbf{w} \right |^{2}p_{k_{0}}^{\star})\\
\Leftrightarrow&S+\left | \mathbf{u}_{k_{0}}^{\rm{H}}\mathbf{H}_{k_{0}}\mathbf{w} \right |^{2}p_{k_{0}}^{\star} > \sqrt{2^{r_{k_{0}}}S(S+\left | \mathbf{u}_{k_{0}}^{\rm{H}}\mathbf{H}_{k_{0}}\mathbf{w} \right |^{2}p_{k_{0}}^{\star})}\\
\Leftrightarrow&\delta>0\\
\end{aligned}
\end{equation}

Then, we calculate the achievable rates of the users. As we have $p_{k}=p_{k}^{\star}~(k=1,2,\cdots,k_{0}-2,k_{0}+1,\cdots,K)$, it is easy to verify that
\begin{equation}
R_{k}=R_{k}^{\star} \geq r_{k}.~(k=1,2,\cdots,k_{0}-2,k_{0}+1,\cdots,K)
\end{equation}
According to $\delta>0$, we have
\begin{equation}
\begin{aligned}
R_{k_{0}-1}&=\log_{2}(1+ \frac{\left | \mathbf{u}_{k_{0}-1}^{\rm{H}}\mathbf{H}_{k_{0}-1}\mathbf{w} \right |^{2}p_{k_{0}-1}}{\left | \mathbf{u}_{k_{0}-1}^{\rm{H}}\mathbf{H}_{k_{0}-1}\mathbf{w} \right |^{2}\sum \limits_{m=1}^{k_{0}-2}p_{m}+\sigma^{2}})\\
&=\log_{2}(1+ \frac{\left | \mathbf{u}_{k_{0}-1}^{\rm{H}}\mathbf{H}_{k_{0}-1}\mathbf{w} \right |^{2}(p_{k_{0}-1}^{\star}+\delta)}{\left | \mathbf{u}_{k_{0}-1}^{\rm{H}}\mathbf{H}_{k_{0}-1}\mathbf{w} \right |^{2}\sum \limits_{m=1}^{k_{0}-2}p_{m}^{\star}+\sigma^{2}})\\
&>\log_{2}(1+ \frac{\left | \mathbf{u}_{k_{0}-1}^{\rm{H}}\mathbf{H}_{k_{0}-1}\mathbf{w} \right |^{2}p_{k_{0}-1}^{\star}}{\left | \mathbf{u}_{k_{0}-1}^{\rm{H}}\mathbf{H}_{k_{0}-1}\mathbf{w} \right |^{2}\sum \limits_{m=1}^{k_{0}-2}p_{m}^{\star}+\sigma^{2}})\\
&=R_{k_{0}-1}^{\star} \geq r_{k_{0}-1},
\end{aligned}
\end{equation}
and according to the expression of $\delta$, we have
\begin{equation}
\begin{aligned}
R_{k_{0}}&=\log_{2}(1+ \frac{\left | \mathbf{u}_{k_{0}}^{\rm{H}}\mathbf{H}_{k_{0}}\mathbf{w} \right |^{2}p_{k_{0}}}{\left | \mathbf{u}_{k_{0}}^{\rm{H}}\mathbf{H}_{k_{0}}\mathbf{w} \right |^{2}\sum \limits_{m=1}^{k_{0}-1}p_{m}+\sigma^{2}})\\
&=\log_{2}(1+ \frac{\left | \mathbf{u}_{k_{0}}^{\rm{H}}\mathbf{H}_{k_{0}}\mathbf{w} \right |^{2}(p_{k_{0}}^{\star}-\delta)}{\left | \mathbf{u}_{k_{0}}^{\rm{H}}\mathbf{H}_{k_{0}}\mathbf{w} \right |^{2}\sum \limits_{m=1}^{k_{0}-1}p_{m}^{\star}+\sigma^{2}})\\
&=\log_{2}(\frac{\sqrt{2^{r_{k_{0}}}S(S+\left | \mathbf{u}_{k_{0}}^{\rm{H}}\mathbf{H}_{k_{0}}\mathbf{w} \right |^{2}p_{k_{0}}^{\star})}}{S})\\
&=\log_{2}\sqrt{2^{r_{k_{0}}}(1+ \frac{\left | \mathbf{u}_{k_{0}}^{\rm{H}}\mathbf{H}_{k_{0}}\mathbf{w} \right |^{2}p_{k_{0}}^{\star}}{\left | \mathbf{u}_{k_{0}}^{\rm{H}}\mathbf{H}_{k_{0}}\mathbf{w} \right |^{2}\sum \limits_{m=1}^{k_{0}-1}p_{m}^{\star}+\sigma^{2}})}\\
&=\frac{R_{k_{0}}^{\star}+r_{k_{0}}}{2}>r_{k_{0}}.
\end{aligned}
\end{equation}

Based on Lemma 2, when $p_{k}=p_{k}^{\star}~(k=1,2,\cdots,k_{0}-2,k_{0}+1,\cdots,K)$, $R_{\rm{sum}}$ is decreasing for $p_{k_{0}}$. Due to $p_{k_{0}}=p_{k_{0}}^{\star}-\delta<p_{k_{0}}^{\star}$, we have
\begin{equation}
R_{\rm{sum}}>R_{\rm{sum}}^{\star},
\end{equation}
which means that under the parameter settings of $\{p_{k}\}$, the minimum rate constraints for all the users are satisfied, and meanwhile the ASR increases. It is in contrast to the assumption that $\{p_{k}^{\star}\}$ is optimal. Thus, we have $R_{k}^{\star}=r_{k}~(2\leq k\leq K)$.

Finally, solve the following equation set
\begin{equation}
\left \{
\begin{aligned}
&R_{k}=r_{k},~(2\leq k\leq K)\\
&\sum\limits_{k=1}^{K}p_{k}=P.
\end{aligned}
\right.
\end{equation}
We can obtain that the optimal power allocation of Problem \eqref{eq_problem3} is given by \eqref{opt_power}.

%\bibliographystyle{IEEEtran} % use IEEEtran.bst style
%\bibliography{IEEEabrv,Xiao60GHz,Xiao5GnNOMA}

\begin{thebibliography}{10}
\providecommand{\url}[1]{#1}
\csname url@samestyle\endcsname
\providecommand{\newblock}{\relax}
\providecommand{\bibinfo}[2]{#2}
\providecommand{\BIBentrySTDinterwordspacing}{\spaceskip=0pt\relax}
\providecommand{\BIBentryALTinterwordstretchfactor}{4}
\providecommand{\BIBentryALTinterwordspacing}{\spaceskip=\fontdimen2\font plus
\BIBentryALTinterwordstretchfactor\fontdimen3\font minus
  \fontdimen4\font\relax}
\providecommand{\BIBforeignlanguage}[2]{{%
\expandafter\ifx\csname l@#1\endcsname\relax
\typeout{** WARNING: IEEEtran.bst: No hyphenation pattern has been}%
\typeout{** loaded for the language `#1'. Using the pattern for}%
\typeout{** the default language instead.}%
\else
\language=\csname l@#1\endcsname
\fi
#2}}
\providecommand{\BIBdecl}{\relax}
\BIBdecl

\bibitem{andrews2014will}
J.~G. Andrews, S.~Buzzi, W.~Choi, S.~V. Hanly, A.~Lozano, A.~C. Soong, and
  J.~C. Zhang, ``What will {5G} be?'' \emph{IEEE Journal on Selected Areas in
  Communications}, vol.~32, no.~6, pp. 1065--1082, Apr. 2014.

\bibitem{Ding2017survNOMA}
Z.~Ding, X.~Lei, G.~K. Karagiannidis, R.~Schober, J.~Yuan, and V.~K. Bhargava,
  ``A survey on non-orthogonal multiple access for {5G} networks: Research
  challenges and future trends,'' \emph{IEEE Journal on Selected Areas in
  Communications}, vol.~35, no.~10, pp. 2181--2195, Oct. 2017.

\bibitem{Dai2018NOMAsurvey}
L.~Dai, B.~Wang, Z.~Ding, Z.~Wang, S.~Chen, and L.~Hanzo, ``A survey of
  non-orthogonal multiple access for {5G},'' \emph{IEEE Communications Surveys
  Tutorials}, pp. 1--1, 2018.

\bibitem{XiaoM2017survmmWave}
M.~Xiao, S.~Mumtaz, Y.~Huang, L.~Dai, Y.~Li, M.~Matthaiou, G.~K. Karagiannidis,
  E.~Bjornson, K.~Yang, C.~L. I, and A.~Ghosh, ``Millimeter wave communications
  for future mobile networks,'' \emph{IEEE Journal on Selected Areas in
  Communications}, vol.~35, no.~9, pp. 1909--1935, Sept. 2017.

\bibitem{Benjebbour2013ConceptNOMA}
A.~Benjebbour, Y.~Saito, Y.~Kishiyama, A.~Li, A.~Harada, and T.~Nakamura,
  ``Concept and practical considerations of non-orthogonal multiple access
  {(NOMA)} for future radio access,'' in \emph{International Symposium on
  Intelligent Signal Processing and Communication Systems}, Nov. 2013, pp.
  770--774.

\bibitem{Dai2015NOMA5G}
L.~Dai, B.~Wang, Y.~Yuan, S.~Han, C.~l.~I, and Z.~Wang, ``Non-orthogonal
  multiple access for {5G}: solutions, challenges, opportunities, and future
  research trends,'' \emph{IEEE Communications Magazine}, vol.~53, no.~9, pp.
  74--81, Sept. 2015.

\bibitem{Choi2014NOMA}
J.~Choi, ``Non-orthogonal multiple access in downlink coordinated two-point
  systems,'' \emph{IEEE Communications Letters}, vol.~18, no.~2, pp. 313--316,
  February 2014.

\bibitem{niu2015survey}
Y.~Niu, Y.~Li, D.~Jin, L.~Su, and A.~V. Vasilakos, ``A survey of millimeter
  wave communications (mmwave) for {5G}: opportunities and challenges,''
  \emph{Wireless Networks}, vol.~21, no.~8, pp. 2657--2676, Apr. 2015.

\bibitem{rapp2013mmIEEEAccess}
T.~S. Rappaport, S.~Sun, R.~Mayzus, H.~Zhao, Y.~Azar, K.~Wang, G.~N. Wong,
  J.~K. Schulz, M.~Samimi, and F.~Gutierrez, ``Millimeter wave mobile
  communications for {5G} cellular: It will work!'' \emph{IEEE Access}, vol.~1,
  pp. 335--349, 2013.

\bibitem{xiao2017mmWaveFD}
Z.~Xiao, P.~Xia, and X.-G. Xia, ``Full-duplex millimeter-wave communication,''
  \emph{IEEE Wireless Communications Magazine}, vol.~16, Dec. 2017.

\bibitem{andrews2016modeling}
J.~G. Andrews, T.~Bai, M.~N. Kulkarni, A.~Alkhateeb, A.~K. Gupta, and R.~W.
  Heath, ``Modeling and analyzing millimeter wave cellular systems,''
  \emph{IEEE Transactions on Communications}, vol.~65, no.~1, pp. 403--430,
  Jan. 2017.

\bibitem{Ding2017random}
Z.~Ding, P.~Fan, and H.~V. Poor, ``Random beamforming in millimeter-wave {NOMA}
  networks,'' \emph{IEEE Access}, vol.~5, pp. 7667--7681, Feb. 2017.

\bibitem{Cui2018mmWaveNOMA}
J.~Cui, Y.~Liu, Z.~Ding, P.~Fan, and A.~Nallanathan, ``Optimal user scheduling
  and power allocation for millimeter wave {NOMA} systems,'' \emph{IEEE
  Transactions on Wireless Communications}, vol.~17, no.~3, pp. 1502--1517,
  March 2018.

\bibitem{Daill2017}
B.~Wang, L.~Dai, Z.~Wang, N.~Ge, and S.~Zhou, ``Spectrum and energy efficient
  beamspace {MIMO-NOMA} for millimeter-wave communications using lens antenna
  array,'' \emph{IEEE Journal on Selected Areas in Communications}, vol.~35,
  no.~10, pp. 2370--2382, Oct. 2017.

\bibitem{Wu2017hybridBF}
W.~Wu and D.~Liu, ``Non-orthogonal multiple access based hybrid beamforming in
  {5G} {mmWave} systems,'' in \emph{2017 IEEE 28th Annual International
  Symposium on Personal, Indoor, and Mobile Radio Communications (PIMRC)}, Oct
  2017, pp. 1--7.

\bibitem{Zhang2017mmWaveMIMONOMA}
D.~Zhang, Z.~Zhou, C.~Xu, Y.~Zhang, J.~Rodriguez, and T.~Sato, ``Capacity
  analysis of {NOMA} with {mmWave} massive {MIMO} systems,'' \emph{IEEE Journal
  on Selected Areas in Communications}, vol.~35, no.~7, pp. 1606--1618, July
  2017.

\bibitem{xiao2018mmWaveNOMA}
Z.~Xiao, L.~Zhu, J.~Choi, P.~Xia, and X.~G. Xia, ``Joint power allocation and
  beamforming for non-orthogonal multiple access {(NOMA)} in {5G} millimeter
  wave communications,'' \emph{IEEE Transactions on Wireless Communications},
  vol.~17, no.~5, pp. 2961--2974, May 2018.

\bibitem{Zhu2018UplinkNOMA}
L.~Zhu, J.~Zhang, Z.~Xiao, X.~Cao, D.~O. Wu, and X.-G. Xia, ``Joint power
  control and beamforming for uplink non-orthogonal multiple access in {5G}
  millimeter-wave communications,'' \emph{IEEE Transactions on Wireless
  Communications}, 2018 (Early Access).

\bibitem{Xiao2018UserFairnessNOMA}
Z.~Xiao, L.~Zhu, D.~O. Wu, and X.-G. Xia, ``User fairness non-orthogonal
  multiple access {(NOMA) for 5G} millimeter-wave communications with analog
  beamforming,'' \emph{IEEE Transactions on Wireless Communications}, 2018.
  (Submitted).

\bibitem{xiao2016codebook}
Z.~Xiao, T.~He, P.~Xia, and X.-G. Xia, ``Hierarchical codebook design for
  beamforming training in millimeter-wave communication,'' \emph{IEEE
  Transactions on Wireless Communications}, vol.~15, no.~5, pp. 3380--3392, May
  2016.

\bibitem{xiao2017codebook}
Z.~Xiao, P.~Xia, and X.-G. Xia, ``Codebook design for millimeter-wave channel
  estimation with hybrid precoding structure,'' \emph{IEEE Transactions on
  Wireless Communications}, vol.~16, no.~1, pp. 141--153, Jan. 2017.

\bibitem{peng2015enhanced}
Y.~Peng, Y.~Li, and P.~Wang, ``An enhanced channel estimation method for
  millimeter wave systems with massive antenna arrays,'' \emph{IEEE
  Communications Letters}, vol.~19, no.~9, pp. 1592--1595, Sept. 2015.

\bibitem{wang2015multi}
P.~Wang, Y.~Li, L.~Song, and B.~Vucetic, ``Multi-gigabit millimeter wave
  wireless communications for {5G}: from fixed access to cellular networks,''
  \emph{IEEE Communications Magazine}, vol.~53, no.~1, pp. 168--178, Jan. 2015.

\bibitem{Lee2014exploiting}
J.~Lee, G.-T. Gil, and Y.~H. Lee, ``Exploiting spatial sparsity for estimating
  channels of hybrid {MIMO} systems in millimeter wave communications,'' in
  \emph{IEEE Global Communications Conference}.\hskip 1em plus 0.5em minus
  0.4em\relax IEEE, 2014, pp. 3326--3331.

\bibitem{Gao2016ChannelEst}
Z.~Gao, C.~Hu, L.~Dai, and Z.~Wang, ``Channel estimation for millimeter-wave
  massive {MIMO} with hybrid precoding over frequency-selective fading
  channels,'' \emph{IEEE Communications Letters}, vol.~20, no.~6, pp.
  1259--1262, June 2016.

\bibitem{alkhateeb2014channel}
A.~Alkhateeb, O.~El~Ayach, G.~Leus, and R.~Heath, ``Channel estimation and
  hybrid precoding for millimeter wave cellular systems,'' \emph{IEEE Journal
  of Selected Topics in Signal Processing}, vol.~8, no.~5, pp. 831--846, Oct.
  2014.

\bibitem{Rahmat-Samii2003PSO}
Y.~Rahmat-Samii, D.~Gies, and J.~Robinson, ``Particle swarm optimization (pso):
  A novel paradigm for antenna designs,'' \emph{URSI Radio Science Bulletin},
  vol. 2003, no. 306, pp. 14--22, Sept. 2003.

\bibitem{Robinson2004PSO}
J.~Robinson and Y.~Rahmat-Samii, ``Particle swarm optimization in
  electromagnetics,'' \emph{IEEE Transactions on Antennas and Propagation},
  vol.~52, no.~2, pp. 397--407, Feb. 2004.

\bibitem{fukuyama2008fundamentals}
Y.~Fukuyama, ``Fundamentals of particle swarm optimization techniques,''
  \emph{Modern Heuristic Optimization Techniques: Theory and Applications to
  Power Systems}, pp. 71--87, 2008.

\end{thebibliography}


% that's all folks%
\end{document}